\documentclass[10pt,twocolumn]{IEEEtran}%

\usepackage{amsmath}
\usepackage{amsfonts}
\usepackage{amssymb}
\usepackage{graphicx}
\usepackage{psfrag}
\usepackage{cite}
\usepackage{algorithm}
\usepackage{color}
\usepackage{float}
\usepackage{epsfig,array,multicol,verbatim,algorithmic}%
\IEEEoverridecommandlockouts
\newtheorem{theorem}{\bf Theorem}
\newtheorem{lemma}{\bf Lemma}
\newtheorem{property}{\bf Property}

\newtheorem{definition}{\bf Definition}
\newtheorem{convention}{\bf Convention}

\newlength{\aligntop}
\newlength{\alignbot}
\makeatletter
\renewenvironment{align}{%
  \vspace{\aligntop}
  \start@align\@ne\st@rredfalse\m@ne
}{%
  \math@cr \black@\totwidth@
  \egroup
  \ifingather@
    \restorealignstate@
    \egroup
    \nonumber
    \ifnum0=`{\fi\iffalse}\fi
  \else
    $$%
  \fi
  \ignorespacesafterend%
  \vspace{\alignbot}\par\noindent
}
\newcommand{\mysection}[1]{\vspace*{-0.47em}\section{#1}\vspace*{-0.49em}}

\begin{document}

\title{Network Formation Games Among Relay Stations in Next Generation Wireless Networks}
\author{Walid Saad, Zhu Han, Tamer Ba\c{s}ar, M\'{e}rouane Debbah, and Are Hj{\o}rungnes \vspace{-1.1cm} \thanks{
W.~Saad is with the Electrical Engineering Department, Princeton University, Princeton, N J, USA email:\texttt{saad@princeton.edu}.
Z.~Han is with Electrical and Computer Engineering Department, University of Houston, Houston, Tx, USA, email: \texttt{zhan2@mail.uh.edu}. T. Ba\c{s}ar is with the Coordinated Science Laboratory, University of Illinois at Urbana Champaign, IL, USA, email: \texttt{basar1@illinois.edu}. M.~Debbah is the Alcatel-Lucent chair, SUPELEC, Paris, France e-mail:
\texttt{merouane.debbah@supelec.fr}. A. Hj{\o}rungnes is
with the UNIK Graduate University Center, University of Oslo, Oslo, Norway, e-mails: \texttt{arehj@unik.no} This work was supported by the Research Council of Norway through the projects 183311/S10, 176773/S10, and 197565/V30,  by NSF grants CNS-0905556 and CNS-0910461, and, in part through the University of Illinois, by AFOSR MURI FA 9550- 10- 1 -0573.}}
\maketitle

\begin{abstract}
The introduction of relay station~(RS) nodes is a key feature in next generation wireless networks such as 3GPP's long term evolution advanced~(LTE-Advanced), or the forthcoming IEEE 802.16j WiMAX standard. This paper presents, using game theory, a novel approach for the formation of the tree architecture that connects the RSs and their serving base station in the \emph{uplink} of the next generation wireless multi-hop systems. Unlike existing literature which mainly focused on performance analysis, we propose a distributed algorithm for studying the \emph{structure} and \emph{dynamics} of the network. We formulate a network formation game among the RSs whereby each RS aims to maximize a cross-layer utility function that takes into account the benefit from cooperative transmission, in terms of reduced bit error rate, and the costs in terms of the delay due to multi-hop transmission. For forming the tree structure, a distributed myopic algorithm is devised. Using the proposed algorithm, each RS can individually select the path that connects it to the BS through other RSs while optimizing its utility. We show the convergence of the algorithm into a Nash tree network, and we study how the RSs can adapt the network's topology to environmental changes such as mobility or the deployment of new mobile stations. Simulation results show that the proposed algorithm presents significant gains in terms of average utility per mobile station which is at least $17.1\%$ better relatively to the case with no RSs and reaches up to $40.3\%$ improvement compared to a nearest neighbor algorithm (for a network with $10$~RSs). The results also show that the average number of hops does not exceed $3$ even for a network with up to $25$ RSs.
\end{abstract}
\vspace{-0.1cm}
\mysection{Introduction}
Cooperation has recently emerged as a novel networking paradigm that can improve the performance of wireless communication networks at different levels. For instance, in order to mitigate the fading effects of the wireless channel, several nodes or relays can cooperate with a given source node in the transmission of its data to a far away destination, thereby, providing spatial diversity gains for the source node without the burden of having several antennas physically present on the node. This class of cooperation is commonly referred to as cooperative communications \cite{LT00}. It has been demonstrated that by deploying one or multiple relays \cite{LT00,SL00,HB00} a significant performance improvement can be witnessed in terms of throughput, bit error rate, capacity, or other metrics. In this regard, existing literature studied various aspects of cooperative transmission such as resource allocation \cite{ZH00}, or link-level performance assessment \cite{LT00,SL00,HB00}. Consequently, due to this performance gain that cooperative communications can yield in a wireless network, recently, the incorporation of relaying into next generation wireless networks has been proposed. In this context, the deployment of relay station~(RS) nodes, dedicated for cooperative communications, is a key challenge in next generation networks such as 3GPP's long term evolution advanced~(LTE-Advanced) \cite{LTE} or the forthcoming IEEE 802.16j WiMAX standard \cite{ST00}.

For an efficient deployment of RSs in next generation networks, several key technical challenges need to be addressed at both the uplink and downlink levels. For the downlink of 802.16j networks, in \cite{LS00}, the authors study the optimal placement of one RS which maximizes the total rate of transmission. In \cite{LTE00}, the authors study the capacity gains and the resource utilization in a multi-hop LTE network in the presence of RSs. Further, the performance of different relaying strategies in an LTE-Advanced network is studied in \cite{LTE01}. In \cite{LS01}, the use of dual relaying is studied in the context of 802.16j networks with multiple RSs. Resource allocation and network planning techniques for 802.16j networks in the presence of RSs are proposed in \cite{OTHER01}. Furthermore, the authors in \cite{LTE02} study the possibility of coverage extension in an LTE-Advanced system, through the use of relaying. In \cite{LTE03}, the communication possibilities between the RSs and the base station is studied and a need-basis algorithm for associating the RSs to their serving BS is proposed for LTE-Advanced networks. The possibilities for handover in an LTE network in the presence of RSs are analyzed in \cite{LTE04}. Other aspects of RS deployment in next generation networks are also considered in \cite{OTHER00,OTHER02,RS00,WS01,WS02}.

Although the performance assessment and operational aspects of RS deployment in next generation multi-hop networks such as LTE-Advanced or 802.16j has been thoroughly studied, one challenging area which remains relatively unexplored is the formation of the tree architecture connecting the BS to the RSs in its coverage area. One contribution toward tackling this problem in 802.16j networks has been made in \cite{RS00} through a centralized approach. However, the work in \cite{RS00} does not provide a clear algorithm for the tree formation nor does it consider cooperative transmission or multi-hop delay. In addition, a centralized approach can yield some significant overhead and complexity, namely in networks with a rapidly changing environment due to RS mobility or incoming traffic load. In our previous work \cite{WS01,WS02}, we proposed game theoretical approaches to tackle the formation of a tree structure in an 802.16j network. However, the model in \cite{WS01} does not account for the costs in terms of the delay incurred by multi-hop transmission while \cite{WS02} is limited to delay tolerant VoIP networks and does not account for the effective throughput of the nodes. In order to take into account both the effective throughput and the delays in the network due to the traffic flow (queueing and transmission delay) for generic services, new models and algorithms, inherently different from \cite{WS01,WS02}, are required.

The main contribution of this paper is to study the distributed formation of the network architecture connecting the RSs to their serving base station in next generation wireless systems such as LTE-Advanced or WiMAX 802.16j. Another key contribution is to propose a cross-layer utility function that captures the gains from cooperative transmission, in terms of a reduced bit error rate and improved effective throughput, as well as the costs incurred by multi-hop transmission in terms of delay. For this purpose, we formulate a network formation game among the RSs in next generation networks, and we build a myopic algorithm in which each RS selects the strategy that maximizes its utility. We show that, through the proposed algorithm, the RSs are able to self-organize into a Nash network tree structure rooted at the serving base station. Moreover, we demonstrate how, by periodic runs of the algorithm, the RSs can take autonomous decisions to adapt the network structure to environmental changes such as incoming traffic due to new mobile stations being deployed as well as mobility. Through simulations, we show that the proposed algorithm leads to a performance gain, in terms of average utility per mobile station, of at least $21.5\%$ compared to the case with no RSs and up to $45.6\%$ compared to a nearest neighbor algorithm.

The rest of this paper is organized as follows: Section~\ref{sec:mod} presents the system model and the game formulation. In Section~\ref{sec:form}, we introduce the cross-layer utility model and present the proposed network formation algorithm. Simulation results are presented and analyzed in Section~\ref{sec:sim}. Finally, conclusions are drawn in Section~\ref{sec:conc}.

\mysection{System Model and Game Formulation}\label{sec:prob}
Consider a network of $M$ RSs that can be either fixed, mobile, or nomadic. The RSs transmit their data in the uplink to a central base station~(BS) through multi-hop links, and, therefore, a tree architecture needs to form, in the uplink, between the RSs and their serving BS. Once the uplink network structure forms, mobile stations~(MSs) can hook to the network by selecting a serving RS or directly connecting to the BS. In this context, we consider that the MSs deposit their data packets to the serving RSs using direct transmission. Subsequently, the RSs in the network that received the data from the external MSs, can act as source nodes transmitting the received MS packets to the BS through one or more hops in the formed tree, using cooperative transmission. The considered direct transmission between an MS and its serving RS enables us to consider a tree formation algorithm that can be easily incorporated in a new or existing  wireless networks without the need of coordination with external entities such as the MSs.

To perform cooperative transmission between the RSs and the BS, we consider a decoded relaying multi-hop diversity channel, such as the one in~\cite{HB00}. In this relaying scheme, each intermediate node on the path between a transmitting RS and the BS combines, encodes, and re-encodes the received signal from all preceding terminals before relaying (decode-and-forward). Formally, every MS $k$ in the network constitutes a source of data traffic which follows a Poisson distribution with an average arrival rate $\lambda_k$. With such Poisson streams at the entry points of the network (the MSs), for every RS, the incoming packets are stored and transmitted in a first-in first-out~(FIFO) fashion and we consider that we have the Kleinrock independence approximation \cite[Chap. 3]{BE00} with each RS being an M/D/1 queueing system\footnote{Any other queueing model, e.g., M/M/1, can also be accommodated.}. With this approximation, the total traffic that an RS $i$  receives from the MSs that it is serving is a Poisson process with an average total arrival rate of $\Lambda_i = \sum_{l\in \mathcal{L}_{i}} \lambda_{l}$ where $\mathcal{L}_i$ is the set of MSs served by an RS $i$ of cardinality $|\mathcal{L}_i|=L_i$. Moreover, RS $i$ also receives packets from RSs that are connected to it with a total average rate $\Delta_i$. For these $\Delta_i$ packets (received from other RSs), the sole role of RS $i$ is to relay them to the next hop. In addition, any RS $i$ that has no assigned MSs and no connected RSs ($\mathcal{L}_i = \emptyset$, $\Lambda_i=0$, and $\Delta_i=0$), transmits ``HELLO'' packets, generated with a Poisson arrival rate of $\eta_0$ in order to maintain its link to the BS active during periods of no actual traffic in the network. An illustrative example of this model is shown in Fig.~\ref{fig:ill}.
\begin{figure}[!t]
\begin{center}
\includegraphics[width=80mm]{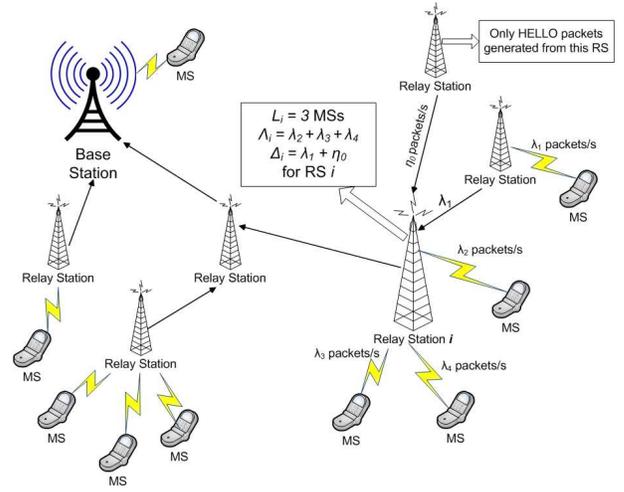}
\end{center}\vspace{-0.5cm}
\caption {A prototype of the uplink tree model.}\vspace{-0.6cm} \label{fig:ill}
\end{figure}

Given this network, the main objective is to provide a formulation that can adequately  model
the interactions between the RSs that seek to form the uplink multi-hop tree architecture. For this purpose, we refer to the analytical framework of network formation games \cite{DM00,TJ00,JD00,TUT00}. Network formation games constitute a subclass of problems which involve a number of independent decisions makers (players) that interact in order to form a suited graph that connects them. The final network graph $G$ that results from a given network formation game is highly dependent on the goals, objectives, and incentives of every player in the game. Consequently, we model the proposed uplink tree formation problem as a network formation game among the RSs where the result of the interactions among the RSs is a \emph{directed} graph $G(\mathcal{V},\mathcal{E})$ with $\mathcal{V}=\{1,\ldots,M+1\}$ denoting the set of all vertices ($M$ RSs and the BS) that will be present in the graph and $\mathcal{E}$ denoting the set of all edges (links) that connect different pairs of RSs. Each directed link between two RSs $i$ and $j$, denoted $(i,j) \in \mathcal{E}$, corresponds to an uplink traffic flow from RS $i$ to RS $j$. We define the following notion of a path:
\begin{definition}
Given any network graph $G(\mathcal{V},\mathcal{E})$, a \emph{path} between two nodes $i \in \mathcal{V}$ and $j \in \mathcal{V}$ is defined as a sequence of nodes $i_1,\ldots,i_K$ (in $\mathcal{V}$) such that $i_1 = i$, $i_K =  j$ and each directed link $(i_k,i_{k+1}) \in G$ for each $k \in \{1,\ldots,K-1\}$.
\end{definition}

In this paper, we consider solely multi-hop tree (or forest, if some parts of the graph are disconnected) architectures, since such architectures are ubiquitous in next generation networks \cite{ST00,LTE00,LTE01}. In this regard, throughout the paper we adopt the following convention:
\begin{convention}
Each RS $i$ is connected to the BS through at most \emph{one} path, and, thus,  we denote by $q_i$ the path between any RS $i$ and the BS whenever this path exists.
\end{convention}

Finally, we delineate the possible actions or strategies that each RS can take in the proposed network formation game. In this regard, for each RS $i$, the action space consists of the RSs (or the BS) that RS $i$ wants to use as its next hop. Therefore, the strategy of an RS $i$ is to select the link that it wants to form from its available action space. We note that, an RS $i$ cannot connect to an RS $j$ which is already connected to $i$, in the sense that if $(j,i) \in G$, then $(i,j) \notin G$. Hence, for a given graph $G$ that governs the current network architecture, we  let $\mathcal{A}_i= \{j \in \mathcal{V}\setminus\{i\}| (j,i) \in G \}$ denote the set of RSs from which RS $i$ accepted a link $(j,i)$,  and $\mathcal{S}_i = \{(i,j) | j \in \mathcal{V}\setminus(\{i\} \bigcup \mathcal{A}_i)  \}$ denote the set of links corresponding to the nodes (RSs or the BS) with which RS $i$ wants to connect (note that $i$ cannot connect to RSs that are already connected to it, i.e., RSs in $\mathcal{A}_i$).  Accordingly, the strategy of an RS $i$ is to select the link $s_i \in \mathcal{S}_i$ that it wants to form, i.e., choose the RS that it will connect to. Based on Convention~1, an RS can be connected to at most \emph{one} other node in our game so selecting to form a link $s_i$  implicitly implies that RS $i$ will \emph{replace} its previously connected link (if any) with the new link $s_i$. Further, to each selection $s_i$ by an RS $i$ corresponds a path $q_i$ to the BS (if $s_i = \emptyset$, then the RS chooses to be disconnected from the network).\vspace{-0.4cm}

\section{Network Formation Game: Utility Function and Algorithm}\label{sec:form}\vspace{-0.1cm}
\subsection{Cross-layer Utility Function}
Our next step is to define a utility function that can capture the incentives of the RSs to connect to each others. For this, we propose a cross-layer utility function that takes into account the performance measures in terms of the packet success rate~(PSR) as well as the delay induced by multi-hop transmission. Hence, considering any tree network graph $G$, each RS in the network will be given a positive utility for every packet that is transmitted/relayed successfully to the BS out of all the packets that this RS has received from the external MSs. In this regard, every packet transmitted by any RS is subject to a bit error rate~(BER) due to the communication over the wireless channel using one or more hops. For any data transmission between an RS $V_1 \in \mathcal{V}$ to the BS, denoted by $V_{n+1}$, going through $n-1$ intermediate RSs $\{V_{2},\ldots,V_{n}\} \subset \mathcal{V}$, let $N_r$ be the set of all receiving terminals, i.e., $N_r = \{V_2\ldots V_{n+1}\}$ and $N_{r(i)}$ be the set of terminals that transmit a signal received by a node $V_i$. Hence, for an RS $V_{i}$ on the path from the source $V_1$ to the destination $V_{n+1}$, we have $N_{r(i)} = \{V_1,\ldots,V_{i-1}\}$. Therefore, given this notation, the BER achieved at the BS $V_{n+1}$ between a source RS $V_1 \in \mathcal{V}$ that is sending its data to the BS along a path $q_{V_1}=\{V_1,\ldots,V_{n+1}\}$ can be calculated through the  tight upper bound given in \cite[Eq. (10)]{HB00} for the decoded relaying multi-hop diversity channel with Rayleigh fading and BPSK modulation\footnote{The approach in this paper is not restricted to this channel and BPSK signal constellation since the algorithm proposed in the following section can be tailored to accommodate other types of relay channels as well as other modulation techniques.} as follows
\begin{align}\label{eq:ber}
P^e_{q_{V_1}} &\le \; \sum_{N_i\in N_r} \frac{1}{2}\left(\sum_{N_k \in N_{r(i)}}\left[\prod_{\underset{N_j \neq N_k}{N_j\in N_{r(i)}} } \frac{\gamma_{k,i}}{\gamma_{k,i}-\gamma_{j,i}}\right.\right.
\nonumber\\
&\left.\left.\times\left(1-\sqrt{\frac{\gamma_{k,i}}{\gamma_{k,i}+1}}\right)\right]\right).
\end{align}
Here, $\gamma_{i,j}=\frac{P_i\cdot h_{i,j}}{\sigma^2}$ is the average received SNR at node $j$ from node $i$ where $P_i$ is the transmit power of node $i$, $\sigma^2$ the noise variance and $h_{i,j} = \frac{1}{d_{i,j}^\mu}$ is the path loss with $d_{i,j}$ the distance between $i$ and $j$ and $\mu$ the path loss exponent. Finally, for RS $i$ which is connected to the BS through a \emph{direct transmission} path $q_{i}^{d}$ with no intermediate hops, the BER can be given by $P^e_{q_i^d} = \frac{1}{2} \left(1-\sqrt{\frac{\gamma_{i,BS}}{1 + \gamma_{i,BS}}}\right)$ \cite{SL00},\cite{HB00}; where $\gamma_{i,BS}$ is the average received SNR at the BS from RS $i$. Using the BER expression in (\ref{eq:ber}) and by having no channel coding, the PSR $\rho_{i,q_i}$ perceived by an RS $i$ over any path $q_i$ is defined as follows
\begin{align}\label{eq:psr}
\rho_{i,q_i}(G) = (1 - P^e_{q_i})^{B},
\end{align}
where $B$ is the number of bits per packet. The PSR is a function of the network graph $G$ as the path $q_i$ varies depending on how RS $i$ is connected to the BS in the formed network tree structure.

Communication over multi-hop wireless links yields a significant delay due to multi-hop transmission as well as buffering. Therefore,  we let $\tau_{i,q_i}$ denote the average delay over the path $q_i = \{i_1,\ldots,i_k\}$ from an RS $i_1 = i$ to the BS. Finding the exact average delay over a path of consecutive queues is a challenging problem in queueing systems \cite{BE00}. One possible approach for measuring the average delay along a path $q_i$ in a network  with Poisson arrivals at the entry points is to consider the Kleinrock approximation as mentioned in the previous section. In this context, the average delay over any path $q_i$ can be given by \cite[Chap. 3, Eqs.~(3.42), (3.45), and (3.93)]{BE00}
 \begin{equation}\label{eq:del}
\tau_{i,q_i}(G) \! =\!\! \sum_{(i_k,i_{k+1}) \in q_i}\!\!\!\! \left(\frac{\Psi_{i_k,i_{k+1}}}{2\mu_{i_k,i_{k+1}}(\mu_{i_k,i_{k+1}} - \Psi_{i_k,i_{k+1}})} + \frac{1}{\mu_{i_k,i_{k+1}}}\right).
\end{equation}
where $\Psi_{i_k,i_{k+1}} = \Lambda_{i_k} + \Delta_{i_k}$ is the total traffic (packets/s) traversing link $(i_k,i_{k+1}) \in q_i$ between RS $i_k$ and RS $i_{k+1}$ and  originating from the $L_{i_k}$ MSs in the set $\mathcal{L}_{i_k}$ of MSs connected to RS $i_k$ ($\Lambda_{i_k}=\sum_{i \in \mathcal{L}_{i_k}}\lambda_i$) and from all RSs  that are connected to $i_k$ ($\Delta_{i_k}=\sum_{j\in \mathcal{A}_{i_k}}\Lambda_{j}$). The ratio $\frac{1}{\mu_{i_k,i_{k+1}}}$ represents the average transmission time (service time) on link $(i_k,i_{k+1}) \in q_i$  with $\mu_{i_{k},i_{k+1}}$ being the service rate on link $(i_k,i_{k+1})$. This service rate is given by $\mu_{i_{k},i_{k+1}}= \frac{C_{i_k,i_{k+1}}}{B}$ with
$C_{i_k,i_{k+1}} = W \log{(1+\nu_{i_k,i_{k+1}})}$
 the capacity of the direct transmission between RS $i_k$ and RS~$i_{k+1}$, where $\nu_{i_k,i_{k+1}}=\frac{P_{i_k}h_{i_k,i_{k+1}}}{\sigma^2}$ is the received SNR from RS $i_k$ at RS $i_{k+1}$, and $W$ is the bandwidth available for RS $i_k$ which is assumed the same for all RSs in the set of vertices $\mathcal{V}$, without loss of generality. Similar to the PSR, the delay depends on the paths from the RSs to the BS, and, hence, it is a function of the network graph $G$.

A suitable criterion for characterizing the utility in networks where the users' quality of service is sensitive to throughput as well as to delay is the concept of \emph{system power}. In this context, power is defined as the ratio of some power of the throughput and the delay \cite{KL00}. Hence, the concept of power is an attractive notion that allows one to capture the fundamental tradeoff between throughput and delay in the proposed network formation game. In fact, the concept of power has been used thoroughly in the  literature to model applications where there exists a tradeoff between throughput and delay\cite{PW03,PW00,PW01,PW02}. Consequently, given the delay and the PSR, we define the utility of an RS $i$ with $L_i$ connected MSs, as the power achieved by $i$ which is given by
\begin{equation}\label{eq:utilPower}
u_i(G)= \begin{cases} \displaystyle  \frac{\left(\Lambda_{i}\cdot \rho_{i,q_i}(G) \right)^{\beta_i}}{\tau_{i,q_i}(G)^{(1-\beta_i)}}, & \mbox{if } L_i > 0,\\ \displaystyle \frac{\left(\eta_{0} \cdot \rho_{i,q_i}(G)\right)^{\beta_i}}{\tau_{i,q_i}(G)^{(1-\beta_i)}}, &
\mbox{if } L_i = 0, \end{cases}
\end{equation}
where $\tau_{i,q_i}(G)$ is the delay given by (\ref{eq:del}),  $\Lambda_{i}\cdot \rho_{i,q_i}(G)$ represents the effective throughput of RS $i$ and $\beta_i \in (0,1)$ is a tradeoff parameter. The utility in (\ref{eq:utilPower}) can model a general class of services, with each class of service having a different $\beta_i$ which can be chosen individually by the RS. As $\beta_i$ increases, the service becomes more delay tolerant and more throughput demanding.  For an RS $i$, the parameter $\beta_i$ can depend on the requirements of its served MSs.  For example, if each MS connected to RS $i$ requests a different value for $\beta_i$, the RS can select the $\beta_i$ to be equal to the value requested by the MS that is most delay sensitive, i.e., the smallest value requested from all connected MSs. As an alternative, the RS can select a value of $\beta_i$ that is averaged over all the values requested from the MSs.
Note that, unless stated otherwise, throughout the rest of the paper the term ``power'' will refer to the ratio of throughput to delay and not to the transmit power of the RSs or MSs unless clearly stated as ``transmit power''.

Once the RSs form the tree topology, one needs to assess the performance of the MSs in terms of the power achieved by these MSs (considered as MS utility). In order to compute the utility of the MSs, the PSR as well as the delay over the whole transmission from MS to BS must be taken into account. Hence, given the proposed network model in Section~\ref{sec:mod}, for each MS $i \in \mathcal{L}_j$ served by an RS $j$, the PSR is given by
\begin{align}\label{eq:MSPSR}
\zeta_{i,j}(G) = \rho_{i,(i,j)}\cdot\rho_{j,q_j}(G),
\end{align}
where $\rho_{i,(i,j)}$ is the PSR on the direct transmission between MS $i$ and RS $j$ (which does not depend on the existing network graph $G$ between the RSs) and $\rho_{j,q_j}(G)$ is the PSR from RS $j$ to the BS along path $q_j$ given by (\ref{eq:psr}) (the path $q_j$ can be either a multi-hop path or a direct transmission depending on how RS $j$ is connected  in the graph $G$ that governs the RSs' network). Furthermore, for any MS $i\in \mathcal{L}_j$ connected to an RS $j$, the delay for transmitting the data to the BS is given by (\ref{eq:del}) by taking into account, in addition to the delay on the RS's path $q_j$, the data traffic on the link $(i,j)$ between the MS and the RS, i.e., the buffering and transmission delay at the MS level. Having the PSR given by (\ref{eq:MSPSR}) and the delay, the utility of a MS $i$ connected to RS $j$ is given by
\begin{equation}\label{eq:utilpowerMS}
v_i(G)= \displaystyle \frac{\left(\lambda_{i}\cdot \zeta_{i,j}(G) \right)^{\beta_i}}{\tau_{i,q_j}(G)^{(1-\beta_i)}}.
\end{equation}
Note that, the MS and RS utilities in (\ref{eq:utilPower}) and (\ref{eq:utilpowerMS}) are selected to represent the node's power which is a metric that links the effective throughput to the delay. For the RSs, the power in (\ref{eq:utilpowerMS}) is a function of the  of the metrics needed for evaluating the MSs' power since it depends on the MSs traffic and their route to destination (with the RS as origin). The MSs power in (\ref{eq:utilpowerMS}) is, in fact, their QoS metric of interest which depends on the direct MS-RS link in addition to the subsequent path from the RS to the BS (which is completely captured by (\ref{eq:utilPower})). The parameter $\beta_i$ in (\ref{eq:utilpowerMS}) is service-dependent and represents how delay tolerant the service used by a certain MS $i$ is.

Consequently, in this paper (unless stated otherwise) we consider that whenever an MS enters the network, it will connect to the RS which maximizes its utility in (\ref{eq:utilpowerMS}) given the current network topology $G$. This MS assignment is considered fixed as long as the RSs' network does not change, otherwise, the MSs can re-assess their utilities and change their assignment once to adapt to the changes in the RSs' network.

Thus, throughout this paper, we mainly deal with the network formation game among the RSs while considering that the MS assignment is fixed once the MS enters the network. The MSs are, as previously mentioned, considered as external sources of traffic. The main advantages behind devising a network formation scheme that relies mainly on the RSs are as follows:
\begin{enumerate}
\item The RSs are typically nodes (fixed, mobile, or nomadic) that are owned by the network operator and that will always be present in the network (except in cases of failures for example). In contrast, the MSs will typically connect to the network for a limited amount of time and, then, leave the network once their connection ends. For this purpose, devising a network formation algorithm among the RSs  has the advantage that it does not rely on external entities such as the MSs which can be entering and exiting the network at random points in time and whose presence in the network can be brief. Further, an RS-only network formation algorithm, can be incorporated in both existing and newly deployed networks.

    \item Although the studied network formation game is between the RSs, as will be seen in Section~\ref{sec:sim}, a significant performance improvement will be witnessed in terms of MS utility as per (\ref{eq:utilpowerMS}). This is due to the fact that, even though network formation is considered only between the RSs, the utilities defined by the RSs in (\ref{eq:utilPower}) take into the key factors impacting the communication path of each MS (e.g., the traffic of the MS and its overall path to destination, i.e., to the BS), except the direct link from MS to RS which is accounted for in the utility of the MS in (\ref{eq:utilpowerMS}). This design improves the performance of the MSs while the MSs do not need to worry about having any knowledge of the network topology or the structure of the tree. The MSs need only to communicate, via a control/feedback channel with the RSs, to select their serving RS based on (\ref{eq:utilpowerMS}). For example, the MS utility in (\ref{eq:utilpowerMS}) can be computed by the RSs on behalf of any MS requesting a connection and then, it is fed back over the control channel. This scheme for assigning MSs to their serving access point, i.e., RSs in this paper, is a standard and well-known method which is already used and deployed in current networks (e.g., cellular or broadband networks) \cite{ZH00}. In consequence, utilizing an RS-only network formation game provides a performance gain to the MSs and does not require additional changes to the standard operation of these MSs.
    \item The MSs can consist of a heterogenous range of devices with different capabilities ranging from small mobile devices to PDAs, laptops, or smartphones. As a result, involving the MSs in network formation would require programming a broad range of devices to act strategically while making network formation decisions. This process can be quite complex in practice. In contrast, the RSs are, in general, standardized nodes (e.g. IEEE 802.16j or LTE-advanced) and, thus,  allowing them to play a network formation game is more reasonable than in the case where the MSs are also involved in the game. One must also remark that the RSs will generally have better processing capabilities than the MSs.

    \item The model proposed in this paper studies a network formation game between a network of RSs with an external incoming traffic which typically comes from MSs. Nonetheless, this external traffic can also come from content providers or servers that need to select an RS to connect to (through a wired or optical network). Hence, one advantage of the proposed model is that it is general enough to accommodate networks with any type of external traffic whether it comes from MSs, content providers, or other sources.
\end{enumerate}
In summary, by designing an RS-based network formation algorithm we are able to extract interesting performance gains, for the MSs, while requiring little interactions or decision making from the MSs which are often devices with limited capabilities that connect to the network for a relatively short period. Nonetheless, for future work, the model considered in this paper can be extended to jointly considers the strategies of the RSs and the MSs. In particular, when considering both the MSs and RSs as players in a network formation game, we can define an interesting and novel multi-leader multi-follower Stackelberg game for network formation. In this game,  the MSs are considered as leaders, i.e., players who can announce their strategies before the other players, known as followers, i.e., the RSs, make their strategy choices. Although the current paper can constitute a key building block for such a multi-leader multi-follower game, this extended model is out of the scope of this paper and will be the subject of future work.

\subsection{Network Formation Algorithm}\vspace{-0.35em}\label{sec:dyn}
Given the devised utility functions in the previous subsection, the next step in the proposed RSs' network formation game is to find an algorithm that can model the interactions among the RSs that seek to form the network tree structure. First, we show that, for any network formation algorithm, the resulting graph in the proposed game is a connected tree structure as follows:
\begin{property}
The network graph resulting from any network formation algorithm for the proposed RSs game is a \emph{connected} directed tree structure rooted at the BS.
\end{property}
\begin{proof}
Consider an RSs network graph $G$ whereby an RS $i$ is disconnected from the BS, i.e., no path of transmission (direct or multi-hop) exists between $i$ and the BS. In this case, one can see that,   the delay for all the packets at the disconnected RS $i$ is infinite, i.e., $\tau_{i,q_i}(G)=\infty$, and, thus, the corresponding power is $0$ as per the utility function in (\ref{eq:utilPower}). As a result, there is no incentive for any RS in the network to disconnect from the BS since such a disconnection will drastically decrease its utility. Hence,  any network graph $G$ formed using the proposed RSs network formation game is a connected graph and due to Convention~1, this graph is a tree rooted at the BS.
\end{proof}

A direct result of this property is that, if any RS is unable to connect to another suitable RSs for forming a link, this RS will connect to the BS using direct transmission. In this regards, we consider that the initial starting point for our network formation game is a star topology whereby all the RSs are connected directly to the BS, prior to interacting for further network formation decisions.

Whenever an RS $i$ plays a strategy $s_i \in \mathcal{S}_i$  while all the remaining RSs maintain a vector of strategies $\boldsymbol{s}_{-i}$, we let $G_{s_i,\boldsymbol{s}_{-i}}$ denote the resulting network graph.  By inspecting the RS utility in (\ref{eq:utilPower}), one can clearly notice that, whenever an RS $j$ accepts a link, due to the increased traffic that it receives, its utility may decrease as the delay increases. Although each RS $i \in \mathcal{N}$ can play any strategy from its strategy space $\mathcal{S}_i$, there might exist some link $s_i = (i,j) \in \mathcal{S}_i$ where the receiving RS, i.e., RS $j$, does not accept the formation of $s_i$, if this leads to a significant decrease in its utility. In this regard, denoting by $G+s_i$ as the graph $G$ modified when an RS $i$ deletes its current link in $G$ and adds the link $s_i=(i,j)$, we define the concept of a \emph{feasible} strategy as follows:
\begin{definition}
A  strategy $s_i \in \mathcal{S}_i$, i.e., a link $s_i = (i,j)$, is a  \emph{feasible strategy} for an RS $i \in \mathcal{V}$ if and only if $u_{j}(G_{s_i,\boldsymbol{s}_{-i}}+s_i) \ge u_{j}(G_{s_i,\boldsymbol{s}_{-i}}) - \epsilon$ where $\epsilon$ is a small positive number. For any RS $i \in \mathcal{V}$, the set of all feasible strategies is denoted by $\hat{\mathcal{S}}_i \subseteq \mathcal{S}_i$.
\end{definition}

A feasible strategy for an RS $i$ is, thus, a link $s_i=(i,j)$ which the receiving RS $j$ is willing to form with RS $i$. Hence, given a network graph $G$, a feasible strategy for any RS $i \in \mathcal{V}$ is to form a link with an RS among all the RSs that are willing to \emph{accept} a connection from RS $i$ (and not \emph{all} RSs), i.e., a feasible path, which maximizes its utility. On the other hand, any RS $j \in \mathcal{V}$ is willing to accept a connection from any other RS $i \in \mathcal{V}$ as long as the formation of the link $(i,j)$ does not decrease the utility of $j$ by more than $\epsilon$. The main motivation for having $\epsilon >0$ (sufficiently small) is that, in many cases, e.g., when the network has only HELLO packets circulating (no MS traffic), RS $j$ might be willing to accept the formation of a link which can slightly decrease its utility at a given moment, but, as more traffic is generated in the network, this link can entail potential future benefits for RS $j$ stemming from an increased effective throughput (recall that the utility in (\ref{eq:utilPower}) captures the tradeoff between effective throughput and delay).

For any RS $i \in \mathcal{V}$, given the set of feasible strategies $\hat{\mathcal{S}}_i$, we define the  \emph{best response} for RS $i$ as follows~\cite{JD00}.
\begin{definition}
A  strategy $s_i^{*} \in \hat{\mathcal{S}}_i$ is a  \emph{best response} for an RS $i \in \mathcal{V}$ if $u_i(G_{s_i^{*},\boldsymbol{s}_{-i}}) \ge u_i(G_{s_i,\boldsymbol{s}_{-i}}),\ \forall s_i \in \hat{\mathcal{S}}_i$. Thus, the best response for RS $i$ is to select the \emph{feasible} link that maximizes its utility given that the other RSs maintain their vector of feasible strategies $\boldsymbol{s}_{-i}$.
\end{definition}

Subsequently, given the various properties of the RS network formation game, we devise a network formation algorithm based on the feasible best responses of the RSs. For this purpose, first, we consider that the RSs are myopic, such that each RS aims at improving its utility given only the current state of the network without taking into account the future evolution of the network. Developing an optimal myopic network formation algorithm is highly complex since there exists no formal rules for network formation in the literature \cite{DM00}. For instance, depending on the model, utilities, and incentives of the players, different network formation rules can be applied. In fact, the topic of network formation is currently hot in game theory and under a lot of research (\cite{DM00,TJ00,JD00} and references therein). The challenging aspect of this problem stems from the fact that one deals with discrete strategy sets (i.e., forming links) and with the formation of network graphs. Further, when dealing with practical utility functions such as (\ref{eq:utilPower}), the problem becomes more challenging. In this context, the game theoretical literature on network formation games studies various myopic algorithms for different game models with directed and undirected graphs \cite{DM00,TJ00,JD00}. For the network formation game among the RSs, we build a myopic algorithm for network formation inspired from those applied in economics problem (e.g., in \cite{DM00} and \cite{JD00}), but modified to accommodate the specifics of the model studied in this paper. In this regard, we define an algorithm where each round is mainly composed of two phases: a myopic network formation phase and a multi-hop transmission phase.

During myopic network formation, the RSs engage in pairwise interactions, sequentially, in order to make their network formation decisions. In this phase, we consider that the RSs make their decisions sequentially in a random order. In practice, this order can be decided by which RS requests first to form its link. Thus, in the myopic network formation phase, each RS $i$ can select a certain feasible strategy which allows it to improve its payoff. An \emph{iteration} consists of a single sequence of plays during which all $M$ RSs have made their strategy choice to myopically react to the choices of the other RSs. The myopic network formation phase can consist of one or more iterations. In every iteration $t$, during its turn, each RS $i$ chooses to play its best response $s_i^{*} \in \hat{\mathcal{S}}_i$  in order to maximize its utility at each round given the current network graph resulting from the strategies of the other RSs. The feasible best response of each RS can be seen as a \emph{replace} operation, whereby the RS will replace its current link to the BS with another link that maximizes its utility (if such a link is available). Hence, the proposed network formation algorithm is based on the iterative feasible best responses of the RSs.

 When it converges, such an algorithm is guaranteed to reach a network where no RS can improve its utility by changing its current link, i.e., a Nash network, defined as follows for the studied game \cite{JD00}:
\begin{definition}\label{def:nash}
A network graph $G(\mathcal{V},\mathcal{E})$ in which no RS $i$ can improve its utility by a unilateral change in its feasible strategy $s_i \in \hat{\mathcal{S}}_i$  is a \emph{Nash network} in the feasible strategy space $\hat{\mathcal{S}}_i, \ \forall i \in \mathcal{V}$.
\end{definition}

A Nash network is simply the concept of a Nash equilibrium applied to a network formation game. In the proposed game, a Nash network would, thus, be a network where no RS can improve its utility, by unilaterally changing its current link, given the current strategies of all other RSs.

Having an analytical proof for the convergence of the network formation phase of the algorithm, when dealing with practical utilities and discrete network formation strategies is difficult \cite{DM00,BA00}. In fact, in wireless applications, even in classical problems such as power control or peer-to-peer incentives (e.g., see \cite{ZH00,WO00,BACD00,PP06}), it is common to propose best-response algorithms even though no analytical proof is found for them, since such algorithms can, in most cases, converge to a Nash equilibrium (or Nash network in the case of network formation).

The iterative best response algorithm we propose in this paper can, thus, either converge to a Nash network or cycle between a number of networks, in the case of non-convergence. In order to avoid such undesirable cycles, one can introduce additional constraints on the strategies of the RSs such as allowing the RSs to select their feasible best response, not only based on the current network, but also on the history of moves or strategies taken by the other RSs, e.g., in repeated games, this is used to ensure reaching an equilibrium where cooperation is enforced \cite{BA00}. Another example is in coalition formation algorithms where, to ensure convergence to a stable point, one can allow the players to experiment, i.e., to select, based on a given history, a coalition that is not the best for them so as to deviate from a cycling behavior \cite{CO00}. Alternatively, in the non-convergence case, the RSs may be instructed by the network operator to find a mixed-strategy Nash network which is guaranteed to exist~\cite{BA00}.

Motivated by such approaches, in the network formation phase of the proposed algorithm, we allow the RSs to observe the visited networks during the occurrence of network formation. In consequence, whenever an RS, based on its history observation, suspects that a cycling behavior is bound to occur, it can deviate from its feasible best response strategy by selecting, instead, the best response that yields a network which was \emph{not previously visited} (at the end of past iterations) more than a certain number of times. Formally, we define a history function $h^t(G_t^i)$ which represents, for every network $G_t^i$ reached at an iteration $t$ during the turn of an RS $i$, the number of times this graph was visited at the end of iterations $1$ to $t-1$. Further, we define a threshold $\varrho$ (positive integer) for $h^t(G_t^i)$ above which an RS $i$ is no longer interested in visiting this network, since visiting this network may lead to a cyclic behavior. Note that, the history function is assumed to be a common knowledge which the RSs can acquire with little complexity (each RS can be made aware of the graph reached at the end of any iteration $t$ by the BS or neighboring RSs).  Thus, at any iteration $t$, if an RS finds out that, by choosing its feasible best response $s_i^*$, it will yield a network $G_t^i + s_i^*$ such that $h^t(G_t^i + s_i^*) > \varrho$, then, this RS will experiment alternative actions by choosing another feasible strategy $s_i \in \hat{\mathcal{S}}_i$ which improves its utility and does not lead to a network $G_t^i+s_i$ with  $h^t(G_t^i+s_i) > \varrho$. Note that, an RS will always try to use its best response first, without reliance on history and it will only use history once and if needed. A critical value $\hat{\varrho}$ for the threshold $\varrho$ is set by the operator so as to control the behavior of the network formation process. This critical value is used by the RSs, only if $\hat{\varrho} \le \varrho$ In essence, if, during the turn of an RS $i$ at an iteration $t$ there exists a network $\hat{G}_t^i$ that has been visited more than $\hat{\varrho}$ (but less than $\varrho$ times), i.e., $h^t(\hat{G}_t^i) > \hat{\varrho}$, then the RSs are instructed to seek a mixed-strategy Nash network.  The mixed-strategy Nash network is a stable network graph $G$ in which each RS can use a number of links,  with different probabilities, for transmitting its data. This is related to the concept of a mixed-strategy Nash equilibrium~\cite{BA00}. The main advantage of seeking a mixed-strategy Nash network is the fact that this network \emph{always exists} independent of the RSs/MSs locations, the circulating traffic, or wireless channel~\cite{BA00}. In this case, for finding the mixed-strategy Nash network, the RSs can use well-known algorithms from the theory of learning in games such as \emph{fictitious play} or evolutionary approaches~\cite{TL00}.

By using these schemes along with iterative best response, multiple iterations will be run until convergence which is guaranteed by the following theorem:
\begin{theorem}\label{th:one}
Given any initial network graph $G_0$, the myopic network formation phase of the proposed algorithm converges to a final network graph $G_T$ after $T$ iterations.
\end{theorem}
\begin{proof}
Every iteration $t$ of the myopic network formation phase of the proposed algorithm can be seen as a sequence of feasible best responses played by the RSs. In this regard, denoting by $G_t$ the graph reached at the \emph{end of} any iteration $t$, the myopic network formation phase consists of a sequence such as the following (as an example)
\begin{equation}\label{eq:trans}
G_0\rightarrow G_1 \rightarrow G_2 \rightarrow \cdots \rightarrow G_t \rightarrow \cdots
\end{equation}
First consider the case in which during the turn of an RS $i$ at any iteration $t$ there does \emph{not} exist any network $\hat{G}_t^i$ that has been visited more than $\hat{\varrho}$, i.e., $h^t(\hat{G}_t^i) \le \hat{\varrho}$ for any $\hat{G}_t^i$. In this case, at any iteration $t$,  denote by $G_t^i$ as the network reached at the turn of an RS $i$. At this iteration, RS $i$ attempts to either select its feasible best response $s_i^* \in \hat{\mathcal{S}}_i$ if $h(G_t^i + s_i^*) \le \varrho$, or, otherwise, it selects a feasible strategy $s_i \in \hat{\mathcal{S}},\ s_i \neq s_i^*$, which improves its utility and yields a network $G_t^i + s_i$ such that $h(G_t^i + s_i) \le \varrho$. This process continues until finding an iteration where no RS can find any strategy to play (i.e.,  no utility improvement is possible for any RS $i$ using a feasible strategy that does not yield a network which has been visited more than $\varrho$ times). Reaching such an iteration is guaranteed by the fact that the number of spanning trees for any graph is \emph{finite}. As a result, the sequence in (\ref{eq:trans}) will always converge to a final graph $G_T$ after $T$ iterations, irrespective of the initial graph $G_0$.

Further, in the case where, during the turn of an RS $i$ at an iteration $t$ there exists a network $\hat{G}_t^i$ that has been visited more than $\hat{\varrho}$, i.e., $h^t(G_t^i) > \hat{\varrho}$, the RSs will seek a mixed-strategy Nash network. While a detailed treatment of the learning process to find the mixed-strategy Nash network is outside the scope of this paper (the interested reader is referred to~\cite{TL00} for more details), the RSs can apply existing algorithms such as fictitious play or evolutionary approaches in order to find the mixed-strategy Nash network~\cite{TL00}.

As a result, the myopic network formation phase of our proposed algorithm always converges.
\end{proof}

We note that, the case in which during the turn of an RS $i$ at any iteration $t$ there does \emph{not} exist any network $\hat{G}_t^i$ that has been visited more than $\hat{\varrho}$, i.e., $h^t(\hat{G}_t^i) \le \hat{\varrho}$ for any $\hat{G}_t^i$, adding the history constraints to the strategies of the RSs implies that, if the algorithm converges after $T$ iterations to a network $G_T$ and \emph{at the final} iteration $T$ (not at any iteration, only at the final one) there exist an RS $i \in \mathcal{V}$ which excluded a certain strategy $s_i$ which yields a better payoff for RS $i$ but leads to a network $G_T^\prime = G_T + s_i,\ G_T^\prime \neq G_T$ such that  $h^T(G_T^\prime )>\varrho$, then the final network $G_T$ is a \emph{history-induced Nash network} and not a Nash network in feasible strategies as per Definition~\ref{def:nash}. The difference is that, in a history-induced Nash network that is formed after $T$ iterations, no RS can, unilaterally, change its link given that its strategy set \emph{excludes} any strategy that yields a network $G_T^\prime$ such that $h^T(G_T^\prime ) > \varrho$ while in a  Nash network in feasible strategies, as per Definition~\ref{def:nash}, no RS has an incentive to unilaterally change its link given its entire feasible strategy set.  We should stress that, the use of history by an RS $i$ at an iteration $t < T$ does \emph{not} mean that the final outcome will necessarily be a history-induced Nash equilibrium. For instance, an RS $i$ can use history instead of using its feasible best response at an iteration $t < T$ and, then, at iterations $t+1,\ldots,T$, it will once again revert to using its best response strategies, if no need for history arise. In this case, the network can, eventually, still reach a Nash network in feasible strategies as per Definition~\ref{def:nash} (not history-induced) at iteration $T$. As a result, we have the following property:
\begin{lemma}
The final tree structure $G_T$ resulting from the convergence of the proposed algorithm after $T$ iterations is a Nash network in the space of feasible strategies $\hat{\mathcal{S}}_i,\ \forall i \in \mathcal{V}$ as per Definition~\ref{def:nash}, if, at iteration $T$, there does not exist any strategy $s_i \in \hat{\mathcal{S}}_i,\  $ for any RS $i$ such that $h_i(G_T +s_i) > \varrho$ and $u_{i}(G_T +s_i) > u_{i}(G_T)$. Otherwise, the final network is a history-induced Nash network. Alternatively, if, in $G_T$, the RSs use different links with different probabilities, then the network is a mixed-strategy Nash network.
\end{lemma}\vspace{-0.2cm}
\begin{proof}
This lemma is a direct consequence of Theorem~\ref{th:one}. Upon convergence of the algorithm to a network $G_T$ after $T$ iterations, we distinguish two cases. If the final network $G_T$ is such that $h^T(G_T + s_i) > \varrho$  and $u_{i}(G_T +s_i) > u_{i}(G_T)$ where $s_i \in \hat{\mathcal{S}}$ is a feasible strategy of any RS $i \in \mathcal{V}$, then, this network is a history-induced Nash network since the only way that this RS can improve its utility is by revisiting a network that was already left more than $\varrho$ times in the past. Otherwise, since the myopic network formation phase of the proposed algorithm is based on the feasible best responses of the RSs at each iteration $t$ and since an RS that uses history at an iteration $t$, can, eventually, revert back to using its feasible best response strategy at iterations $t+1,\ldots,T$, then the final network $G_T$ is a Nash network in feasible strategies as per Definition~\ref{def:nash}.  (the convergence of a best response algorithm reaches a Nash equilibrium \cite{BA00}). In this Nash network, no RS can improve its utility by unilaterally deviating from its currently selected feasible strategy (with no use of history). Alternatively, in the case in which, during the turn of an RS $i$ at an iteration $t$ there exists a network $\hat{G}_t^i$ that has been visited more than $\hat{\varrho}$, i.e., $h^t(G_t^i) > \hat{\varrho}$, the final network $G_T$ will be a mixed-strategy Nash network in which the RSs use different links with different probabilities.
\end{proof}

Note that the value of $\hat{\varrho}$ can be set by the operator in a way to highlight a certain preference between the mixed-strategy case and the history induced case.  For example, if the operator prefers the mixed-strategy case then it can set $\hat{\varrho} \le \varrho$, otherwise, when the operator prefers the history induced Nash network, it will set $\hat{\varrho} > \varrho$. Whenever no Nash network in feasible strategies is found, this preference captures a tradeoff between stronger stability (mixed-strategy Nash network) or faster convergence (history induced Nash network)..

After the convergence of the network formation phase of the algorithm, the RSs are connected through a tree structure $G_T$ and the second phase of the algorithm begins. This phase represents the actual data transmission phase, whereby the multi-hop network operation occurs as the RSs transmit the data over the existing tree architecture $G_T$.  A summary of the proposed algorithm is given in Table~\ref{tab:algo}.

 \begin{table}[!t]
  \centering
  \caption{
    \vspace*{-0.7em}Proposed network formation algorithm.}\vspace*{-1em}
    \begin{tabular}{p{8cm}}
      \hline
      \textbf{Initial State} \vspace*{.3em} \\
      \hspace*{1em}The starting network is a graph where the RSs are directly connected to the BS (star network).\vspace*{0.01cm}\\
\textbf{The proposed algorithm consists of two phases} \vspace*{.3em}\\
\hspace*{1em}\emph{Phase~I - Myopic Network Formation:}   \vspace*{.1em}\\
\hspace*{2em}\textbf{repeat}\vspace*{.2em}\\
\hspace*{3em}In a random but sequential order, the RSs engage in a network\vspace*{.2em}\\
\hspace*{3em}formation game.\vspace*{.2em}\\
\hspace*{3em}a) In every iteration $t$ of Phase~I, each RS $i$ plays its best\vspace*{.2em}\\
\hspace*{3em}response.\vspace*{.2em}\\
\hspace*{4em}a.1) Whenever, during the turn of an RS $i$ at an iteration $t$\vspace*{.2em}\\
\hspace*{4em}there exists a network $\hat{G}_t^i$ that has been visited more than\vspace*{.2em}\\
\hspace*{4em}$\hat{\varrho}$, i.e., $h^t(G_t^i) > \hat{\varrho}$ the RSs will seek a mixed-strategy Nash\vspace*{.2em}\\
\hspace*{4em}network using well-known techniques such as fictitious\vspace*{.2em}\\
\hspace*{4em}play or evolutionary games~\cite{TL00}.\vspace*{.2em}\\
\hspace*{4em}a.2) Otherwise, each RS $i$ maximizes its utility by playing\vspace*{.2em}\\
\hspace*{4em}its feasible best response $s_i^{*} \in \hat{\mathcal{S}}_i \setminus  \mathcal{S}_{\mathcal{G}_t}$ (with $\mathcal{G}_t$ being the\vspace*{.2em}\\
\hspace*{4em}set of all graphs visited at the end of iterations $1$ till $t-1$).\vspace*{.2em}\\
\hspace*{3em}b) For a.2), the best response $s_i^{*}$ of each RS is a \emph{replace}\vspace*{.2em}\\
\hspace*{3em}operation through which an RS $i$ splits from its current parent\vspace*{.2em}\\
\hspace*{3em}RS and replaces it with a new RS that maximizes its utility,\vspace*{.2em}\\
\hspace*{3em}given that this new RS \emph{accepts} the formation of the link.\vspace*{.2em}\\
\hspace*{2em}\textbf{until} convergence to a final Nash tree $G_T$ after $T$ iterations.\vspace*{.2em}\\
\hspace*{1em}\emph{Phase~II - Multi-hop Transmission:}   \vspace*{.1em}\\
\hspace*{3em}During this phase, data transmission from the MSs occurs using\vspace*{.2em}\\
\hspace*{3em}the formed network tree structure $G_T$.\vspace*{.2em}\\
\textbf{For changing environments (e.g. due to mobility or the deployment of new MSs), multiple rounds of this algorithm are run \emph{periodically} every time period $\theta$, allowing the RSs to adapt the network topology.} \vspace*{.5em}\\
   \hline
    \end{tabular}\label{tab:algo}\vspace{-0.7cm}
\end{table}

Furthermore, as the RSs can engage in the myopic network formation phase prior to any MS deployment, we consider the following convention throughout the rest of this paper:
\begin{convention}\label{conv:t}
At the beginning of all time, once the operator deploys the network, the RSs engage in the network formation game by taking into account their utilities in terms of HELLO packets, prior to any mobility or presence of MSs.
\end{convention}
The main motivation behind Convention~\ref{conv:t} is that the RSs can form an initial tree structure which shall be used by any MSs that will be deployed in the network. If any adaptation to this structure is needed, periodic runs of the proposed algorithm can occur as discussed further in this section.

The proposed algorithm can be implemented in a distributed way within any next generation wireless multi-hop network, with a little reliance on the BS. For instance, the sole role of the BS in the proposed network formation algorithm is to inform the RSs of the graphs reached during past iterations, if needed, over a control channel. Due to the fact that the number of RSs within the area of a single BS is small when compared with the number of MSs, the signalling and overhead for this information exchange between the BS and the RSs is minimal. Beyond this, the algorithm relies on distributed decisions taken by the RSs. Within every iteration $t$, during its turn, each RS can engage in pairwise negotiations with the surrounding RSs in order to find its best response, among the set of feasible strategies and given the graphs that were reached in previous iterations.

Note that, although a fully centralized approach can also be used for implementing the proposed algorithm, the need for a distributed solution is desirable as it has several advantages. First, although the number of RSs compared to that of the MSs is generally small, the proposed approach can also apply to the case where the RSs are replaced with relay nodes, whose number can be significantly large hence motivating a distributed approach. Second, a distributed approach can reduce the communication overhead at the BS, notably when the BS controls a large area in which the RSs are deployed in order to alleviate the communication overhead at the BS by communicating (instead of the BS) with some of the MSs. Also, in this case,  a centralized approach can require the BS to communicate with all of its RSs and update the network whenever needed, hence, increasing the signalling in the network and the computational load on the BS. Further, a distributed network formation game is more robust to increased delays at the BS (e.g., due to traffic received from non-MS sources such as content providers), failures, as well as to malicious attacks. This is due to the fact that, unlike a centralized approach, the distributed network formation game does not rely on a single controller such as the BS which, if compromised (due to malicious attacks or failures), can lead to a failure at the level of the entire network. Finally, although the current paper mainly deals with networks having a standard infrastructure (e.g., WiMAX or LTE-Advanced), the approach can equally apply in an ad hoc network where the relay stations are, in fact, relay nodes and the central base station is a common receiver for these nodes. In such a case, it is desirable that the nodes take their own decisions on how and where to route and transmit their traffic. Such distributed decision making is also of interest in an infrastructure-based network, whenever the RSs are mobile, or when the number and identity of the RSs can vary over time. In such a case, it is difficult (and undesirable) for a centralized entity to keep track of the variations in the network. For these reasons, a distributed approach for network formation is well-motivated.

The worst case complexity for implementing Step (a) in Phase~II of the algorithm in Table~\ref{tab:algo}, i.e., selecting the feasible best response (finding a suited partner) for any RS $i$ is $O(M)$ where $M$ is the total number of RSs. In practice, the complexity is much smaller as the RSs do not negotiate with the RSs that are connected to them, nor with the RSs that can lead to a graph visited at previous iterations. We stress that the complexity to find the best response in the proposed game is comparable with some of the most popular game theoretic approaches that are used in the literature when tackling problems such as power control or resource allocation (see \cite{ZH00} for a thorough overview on such approaches) in which finding the best response can yield a non-negligible and sometimes exponential complexity.  In order to evaluate its utility while searching for the best response, each RS can easily acquire the BER and an estimate of the delay that each neighbor can provide. As a result, each RS $i$ can take an individual decision to select the link $s_i^{*}$ that can maximize its utility. The signaling required for gathering this information can be minimal as each RS can measure its current channel towards the BS as well as the flowing traffic and feed this information back to any RS that requests it during the pairwise negotiations. In dynamically changing environments, following the formation of the initial tree structure as per Convention~\ref{conv:t}, the network formation process is repeated periodically every $\theta$ allowing the RSs to take autonomous decisions to update the topology adapting it to any environmental changes that occurred during $\theta$ such as the deployment of MSs, mobility of the RSs and/or MSs, among others. In fact, engaging in the network formation game periodically rather than continuously reduces the signalling in the network, while allowing the topology to adapt itself to environmental changes.  As the period $\theta$ is chosen to be smaller, the network formation game is played more often, allowing a better adaptation to networks with rapidly changing environments at the expense of extra signalling and overhead.  Note that, when the RSs are mobile, and/or when new MSs are entering and leaving the network, the MSs can also, periodically, change their serving RS, to adapt to this change in the network.

\section{Simulation Results and Analysis}\label{sec:sim}
For simulations, we consider a square area of $3$~km $\times$ $3$~km with the BS at the center. We deploy the RSs and the MSs within this area. The transmit power is set to $50$~mW  for all RSs and MSs, the noise level is $-100$~dBm, and the bandwidth per RS is set to $W=100$~kHz. For path loss, we set the propagation loss to $\mu=3$. We consider a traffic of $64$~kbps, divided into packets of length $B=256$ bits with an arrival rate of $250$ packets/s. For the HELLO packets, we set $\eta_0=1$~packet/s with the same packet length of $B=256$ bits. Unless stated otherwise, we assume that all the RSs and MSs utilize the same tradeoff parameters and its value is set to $\beta_i=\beta=0.7$ (for all RS and MS $i$) to imply services that are slightly delay tolerant. Further, the parameter $\epsilon$ is selected to be equal to $1\%$ of any RS's current utility, i.e., an RS accepts the formation of a link if its utility does not decrease by more than $1\%$ of its current value. Finally, we set $\varrho=1$ and $\hat{\varrho} > 1$.

In Fig.~\ref{fig:snapshot}, we randomly deploy $M=10$ RSs within the area of the BS. The network starts with an initial star topology with all the RSs connected directly to the BS. Prior to the deployment of MSs (in the presence of HELLO packets only), the RSs engage in the proposed network formation algorithm and converge to the final Nash network structure shown by the solid lines in Fig.~\ref{fig:snapshot}. Clearly, the figure shows that through their distributed decisions the RSs select their preferred nearby partners, forming the multi-hop tree structure. Furthermore, we deploy $30$ randomly located MSs in the area, and show how the RSs self-organize and adapt the network's topology to the incoming traffic through the dashed lines in Fig.~\ref{fig:snapshot}. For instance, RS $9$ improves its utility from $370.6$ to $391.2$ by disconnecting from RS $8$ and connecting to RS $6$ instead. This improvement stems from the fact that, although  connecting to RS $8$ provides a better BER for RS $9$, in the presence of the MSs, choosing a shorter path, i.e., less hops through RS $6$, the delay perceived by the traffic of RS $9$ is reduced, hence, improving the overall utility. Moreover, due to the deployment of traffic and the deviation of RS $9$, RS $8$ decides to disconnect from RS $6$ and connect directly to the BS, hence, avoiding the extra delay that exists at RS $6$ when MSs are deployed. Further, in order to send its HELLO packet, RS $7$ finds it beneficial to replace its current link with the congested RS $1$ with a direct link to the BS. In brief, Fig.~\ref{fig:snapshot} summarizes the operation of the proposed adaptive network formation algorithm with and without the presence of external traffic from MSs.
\begin{figure}[!t]
\begin{center}
\includegraphics[width=80mm]{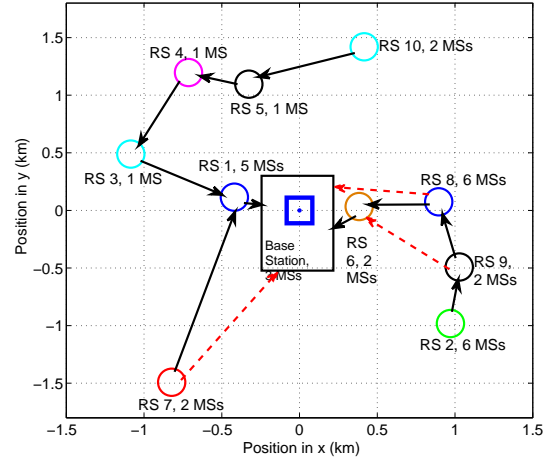}
\end{center}\vspace{-0.6cm}
\caption {Snapshot of a tree topology formed using the proposed network algorithm with $M=10$ RSs before (solid line) and after (dashed line) the random deployment of $30$ MSs.} \label{fig:snapshot}\vspace{-0.6cm}
\end{figure}
\begin{figure}[!t]
\begin{center}
\includegraphics[width=80mm]{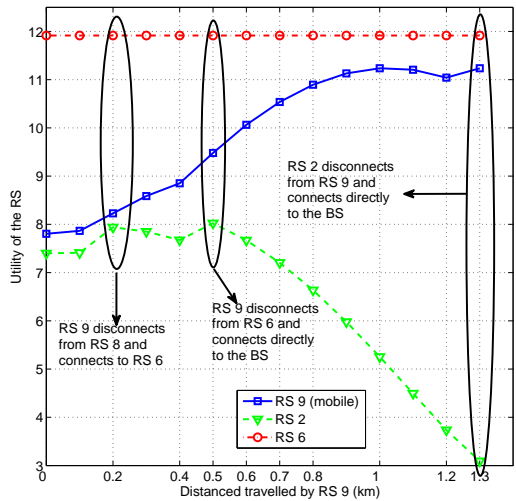}
\end{center}\vspace{-0.6cm}
\caption {Adaptation of the network's tree structure to mobility of the RSs shown through the changes in the utility of RS $9$ of Fig.~\ref{fig:snapshot} as it moves on the x-axis in the negative direction prior to any MS presence.} \label{fig:split}\vspace{-0.3cm}
\end{figure}

In Fig.~\ref{fig:split}, we assess the effect of mobility on the network structure. For this purpose, we consider the network of Fig.~\ref{fig:snapshot} \emph{prior to the deployment of the MSs} and we consider that RS~$9$ is moving horizontally in the direction of the negative x-axis while the other RSs remain static. The variation in the utilities of the main concerned RSs during the mobility of RS $9$ are shown in Fig.~\ref{fig:split}. Once RS~$9$ starts its movement, its utility increases since its distance to its serving RS, RS~$8$, decreases. Similarly, the utility of RS~$2$, served by RS~$9$ also increases. As RS~$9$ moves around $0.2$~km, it finds it beneficial to replace its current link with RS~$8$ and connect to RS~$6$ instead. In this context, RS~$6$ would accept the incoming connection from RS~$9$ since this acceptance does not affect its utility negatively as shown in Fig.~\ref{fig:split} at $0.2$~km. As RS~$9$ pursues its mobility, its utility improves as it gets closer to RS~$6$ while the utility of RS~$2$ decreases since RS~$9$ is distancing itself from it. After moving for a distance of $0.5$~km, RS~$9$ becomes quite close to the BS, and, thus, it maximizes its utility by disconnecting from RS~$6$ and connecting directly to the BS. This action taken by RS~$9$ at $0.5$~km also improves the utility of RS~$2$. Meanwhile, RS~$9$ continues its movement and its utility as well as that of RS~$2$ start to drop as RS~$9$ distances itself from the BS. As soon as RS~$9$ moves for a total of $1.3$~km, RS~$2$ decides to disconnect from RS~$9$ and connect directly to the BS since the direct transmission can provide a better utility at this point. In a nutshell, by inspecting the results of Fig.~\ref{fig:split}, we clearly illustrate how the RSs can take distributed decisions that allow them to self-organize and adapt the topology to mobility.

\begin{figure}[!t]
\begin{center}
\includegraphics[width=80mm]{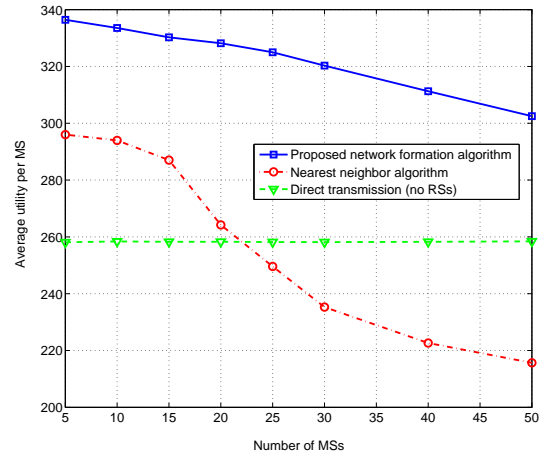}
\end{center}\vspace{-0.6cm}
\caption {Performance assessment of the proposed network formation algorithm, in terms of average utility per MS, for a network having $M=10$~RSs as the number of MSs varies.} \label{fig:perf}\vspace{-0.3cm}
\end{figure}

Further, we have run a variety of statistical simulation results for different network sizes, with each one averaged over around $50,000$ iterations with random positions for the MSs and RSs. Based on these simulations, we first note that, for networks having up to $M=10$~RSs (for any number of MSs), we have encountered \emph{only} Nash networks in feasible strategies as per Definition~\ref{def:nash} and \emph{no} history-induced Nash networks. For networks with $M>10$~RSs, only about $2\%$ of the runs ended up with a history-induced Nash network. Hence, based on these simulations, we can see that, in practical settings, the number of history-induced Nash networks is very small, and, thus, the RSs may reach a Nash network without utilizing their history functions in the decision process.

Subsequently, Fig.~\ref{fig:perf} shows the average achieved utility per MS for a network with $M=10$~RSs as the number of MSs in the network increases. The performance of the proposed network formation algorithm is compared against the direct transmission performance, i.e.,  the case where no RSs exist in the network, as well as a nearest neighbor algorithm whereby each node selects the closest partner to connect to. Note that these schemes are selected for comparison purposes since, to the best of our knowledge, this paper is the first in the literature that deals with distributed tree formation in next-generation networks\footnote{The work in \cite{RS00} studies the tree formation in IEEE 802.16j, however, \cite{RS00} focuses on the messages needed to control the RSs and no algorithm (or QoS metric/utility) for forming the network is actually provided.}. In this figure, we can see that, as the number of MSs in the network increases, the performance of both the proposed algorithm as well as that of the nearest neighbor algorithm decrease. This result is due to the fact that, as more MSs are present in the network, the delay from multi-hop transmission due to the additional traffic increases, and, thus, the average payoff per MS decreases. In contrast, in the case of no RSs, the performance is unaffected by the increase in the number of MSs since no delay exists in the network. We also note that, due to the increased traffic, the performance of the nearest neighbor algorithm drops below that of the direct transmission at around $20$~MSs. Further, Fig.~\ref{fig:perf} shows that, at all network sizes, the proposed network formation algorithm presents a significant advantage over both the nearest neighbor algorithm and the direct transmission case. This performance advantage is of at least $17.1\%$ compared to the direct transmission case (for $50$~MSs) and it reaches up to $40.3\%$ improvement relative to the nearest neighbor algorithm at $50$~MSs.

\begin{figure}[!t]
\begin{center}
\includegraphics[width=80mm]{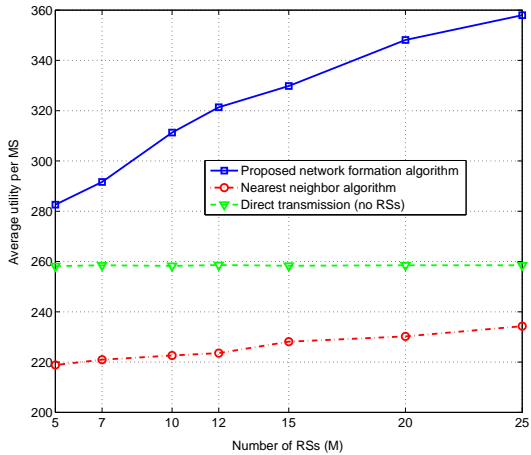}
\end{center}\vspace{-0.6cm}
\caption {Performance assessment of the proposed network formation algorithm, in terms of average utility per MS, for a network having $40$~MSs as the number of RSs $M$ varies.} \label{fig:perfrs}\vspace{-0.3cm}
\end{figure}

The performance of the proposed network formation algorithm is further assessed in Fig.~\ref{fig:perfrs}, where we show the average utility per MS as the number of RSs $M$ in the network varies, for a network with $40$~MSs. Fig.~\ref{fig:perfrs} shows that, as $M$ increases, the performance of the proposed algorithm as well as that of the nearest neighbor algorithm increase. This is due to the fact that, as the number of RSs increase, the possibilities of benefiting from cooperative transmission gains increase, and, thus, the average utility per MS increase. In contrast, for the direct transmission scheme, the performance is constant as $M$ varies, since this scheme does not depend on the number of RSs. Fig.~\ref{fig:perfrs} demonstrates that, at all network sizes, the proposed network formation algorithm presents a significant performance gain reaching, respectively, up to $52.8\%$ and $38.5\%$ relative to the nearest neighbor algorithm and the direct transmission case.


In Fig.~\ref{fig:hops}, we show the average and the average maximum number of hops in the resulting network structure as the number of RSs $M$ in the network increases for a network with $40$~MSs
 (results are averaged over random positions of MSs and RSs). The number of hops shown in this figure represents the hops connecting RSs or the RSs to the BS, without accounting for the MS-RS hop. Fig.~\ref{fig:hops} shows that, as the number of RSs $M$ increases, both the average and the  average maximum number of hops in the tree structure increase. The average and the average maximum number of hops vary, respectively, from $1.85$ and $2.5$ at $M=5$~RSs, up to around $3$ and $5$ at $M=25$. Consequently, as per  Fig.~\ref{fig:hops}, due to the delay cost for multi-hop transmission, both the average and average maximum number of hops increase very slowly with the network size $M$. For instance, one can notice that, up to $20$ additional RSs are needed in order to increase the average number of hops of around $1$ hops and the average maximum number of hops of only around $2$~hops.

\begin{figure}[!t]
\begin{center}
\includegraphics[width=80mm]{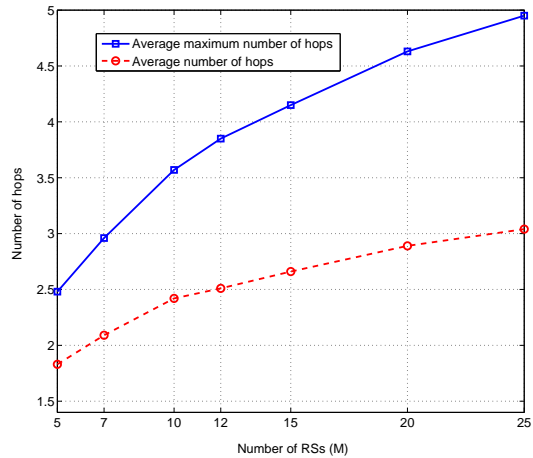}
\end{center}\vspace{-0.6cm}
\caption {Average and average maximum number of hops in the final tree structure for a network with $40$ MSs vs. number of RSs $M$ in the network.} \label{fig:hops}\vspace{-0.3cm}
\end{figure}

Fig.~\ref{fig:conv} shows the average and the maximum number of iterations needed till convergence of the algorithm to the initial network structure prior to the deployment of any MSs, as the size of the network $M$ increases. This figure shows that, as the number of RSs increase, the total number of iterations required for the convergence of the algorithm increases. This result is due to the fact that, as $M$ increases, the cooperation options for every RS increase, and, thus, more actions are required prior to convergence. Fig.~\ref{fig:conv} shows that the average and the maximum number of iterations vary, respectively, from $1.12$ and $2$ at $M=5$~RSs up to $2.9$ and $8$ at $M=25$~RSs. Hence, this result demonstrates that, in average, the speed of convergence of the proposed algorithm is quite reasonable even for relatively large networks. Similar results can be seen for the convergence of the algorithm when MSs are deployed or when the RSs are moving.

\begin{figure}[!t]
\begin{center}
\includegraphics[width=80mm]{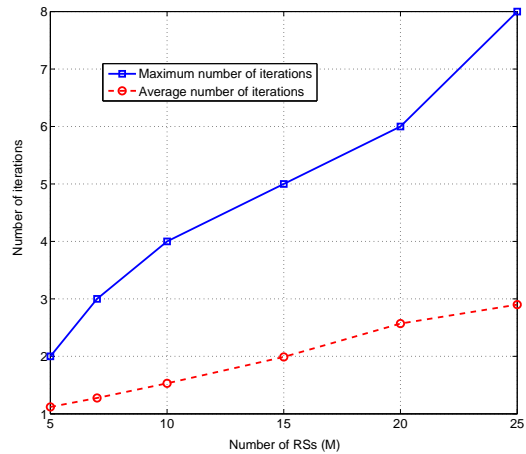}
\end{center}\vspace{-0.6cm}
\caption {Average and maximum number of iterations till convergence vs. number of RSs $M$ in the network.} \label{fig:conv}\vspace{-0.3cm}
\end{figure}

Fig.~\ref{fig:nash} shows the distribution of the total number of Nash networks, over about $50,000$ iterations (network settings), for networks with $40$~MSs for the two cases of $M=5$~RSs and $M=7$~RSs. Each iteration represents different locations for the MSs and RSs. This figure is generated by finding all possible network trees (spanning trees) and counting the number of Nash networks in each case. Fig.~\ref{fig:nash} shows that for all settings, at least one Nash network exists. Further, for $M=5$~RSs, we can see that the number of Nash networks is concentrated in the interval $[1,10]$. In fact, for $M=5$~RSs about $83\%$ of the cases admit between $1$ and $10$ Nash networks with the majority of the network settings having $4$ Nash networks (about $25\%$ of the total cases). However, as the number of RSs increases of $2$, i.e., for $M=7$~RSs, Fig.~\ref{fig:nash} shows that more Nash networks exist and their distribution becomes more balanced over the different intervals. At $M=7$~RSs, about $23\%$ of the cases admit between $11$ and $20$ Nash networks and about $19\%$ admit between $31$ and $40$ Nash networks. For $M=7$~RSs, we can see that about $9\%$ admit more than $100$ Nash networks with the maximum being one case having $1293$ Nash networks. Although, at first glance, this number can look large, it must be noted that $M=5$~RSs and $M=7$~RSs can form, respectively, a total of $125$ and $16,807$ possible trees (this number is given by Cayley's formula which states that a graph with $n$ vertices admits $n^{n-2}$ spanning trees \cite{GRA00}). Thus, relative to the total number of possible network trees, the number of Nash networks is small. Note that, for large networks, finding all possible Nash equilibria and their distribution is computationally intractable since it requires finding all possible networks which grow exponentially with $M$. However, Fig.~\ref{fig:nash} gives a good insight on how this number will vary as the network size grows.

\begin{figure}[!t]
\begin{center}
\includegraphics[width=80mm]{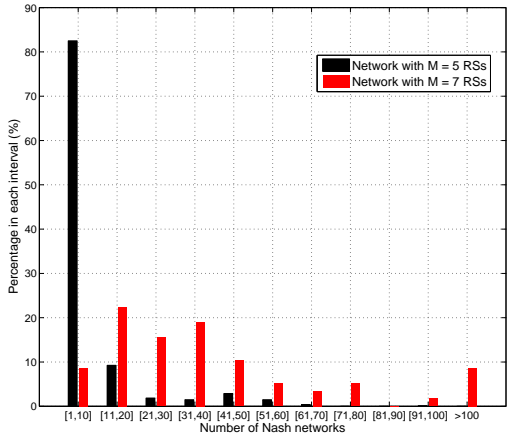}
\end{center}\vspace{-0.6cm}
\caption {Distribution of the number of Nash networks over about
 $50,000$ different network settings (RS and MS locations) for networks with $M=5$~RSs and $M=7$~RSs (with $40$~MSs).} \label{fig:nash}\vspace{-0.3cm}
\end{figure}

In Fig.~\ref{fig:nasheff}, we assess the efficiency of the Nash networks in the proposed model by showing the average utility per MS achieved by the proposed network formation algorithm, a centralized approach that finds the optimal (maximizing the average utility per RS) network tree by exhaustive search, and the Nash network having the least efficiency, i.e., the smallest average utility per MS (worst case Nash network) for a network with $M=5$~RSs as the number of MSs varies. In this figure, we note that the proposed network formation algorithm achieves, at all network sizes, a performance that is comparable to the optimal solution. The performance of the proposed network formation algorithm gets closer to the optimal solution as the network becomes congested. For instance, Fig.~\ref{fig:nasheff} shows that the average utility per MS resulting from network formation is only between $5\%$ (at $10$~MSs) and $2.8\%$ (at $M=50$~MSs) less than the optimal solution. Moreover, this figure provides an insight on the efficiency of the Nash networks resulting from the proposed model through the price of anarchy, which is defined as the ratio between the optimal case and the worst case Nash equilibrium~\cite{ALG00}. Fig.~\ref{fig:nasheff} shows that the price of anarchy is, on the average, about $1.09$. This result shows that the Nash networks resulting from the proposed model are, in general, reasonably efficient as the worst case Nash network has a performance of not less than $9\%$ below the optimal solution.
\begin{figure}[!t]
\begin{center}
\includegraphics[width=80mm]{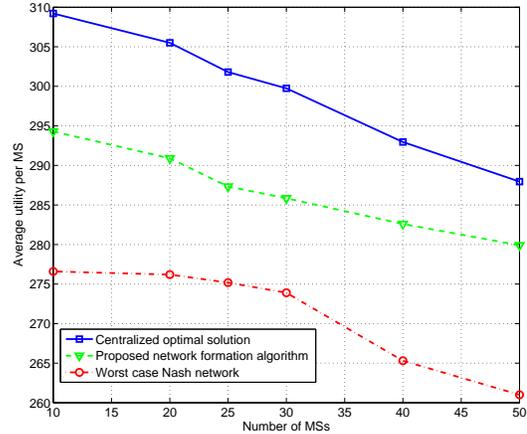}
\end{center}\vspace{-0.6cm}
\caption{Average utility per MS achieved by the proposed algorithm, a centralized approach that finds the optimal network tree by exhaustive search, and the Nash network having the least efficiency, i.e., the smallest average utility per MS (worst case Nash network) for a network with $M=5$~RSs as the number of MSs varies.} \label{fig:nasheff}\vspace{-0.3cm}
\end{figure}

In Fig.~\ref{fig:hopsbeta}, we show the average and the average maximum number of hops for a network with $M=10$~RSs and $40$~MSs as the tradeoff parameter $\beta$ varies (results are averaged over random positions of MSs and RSs). Fig.~\ref{fig:hopsbeta} shows that, as the tradeoff parameter increases, both the average and the average maximum number of hops in the tree structure increase. For instance, the average and the average maximum number of hops vary, respectively, from $1.14$ and $1.75$ at $\beta=0.1$, up to around $2.8$ and around $4$ at $\beta=0.9$. The increase in the number of hops with $\beta$ is due to the fact that, as the network becomes more delay tolerant (larger $\beta$) the possibilities for using multi-hop transmission among the RSs increases. In contrast, as the network becomes more delay sensitive, i.e., for small $\beta$, the RSs tend to self-organize into a tree structure with very small number of hops. For instance, at $\beta=0.1$, the average number of hops is quite close to $1$, which implies that, for highly delay sensitive services, direct transmission from the RSs to the BS, i.e., the star topology, provides, on the average, the best architecture for communication.
\begin{figure}[!t]
\begin{center}
\includegraphics[width=80mm]{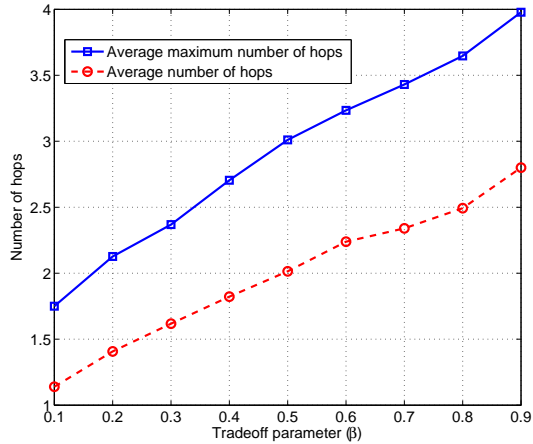}
\end{center}\vspace{-0.6cm}
\caption {Average and average maximum number of hops in the final tree structure for a network with $10$~RSs and $40$ MSs as the tradeoff parameter $\beta$ varies.} \label{fig:hopsbeta}\vspace{-0.3cm}
\end{figure}

In Fig.~\ref{fig:freq},  we show, over a period of $5$ minutes, the average total number of actions taken by all RSs for various velocities of the RSs in a wireless network with $40$~MSs and different number of RSs. The proposed network formation algorithm is repeated by the RSs, periodically, every $\theta=30$~seconds, in order to provide self-adaptation to mobility. As the speed of the RSs increases, the average total number of actions per minute increases for both $M=10$~RSs and $M=20$~RSs. This result corroborates the fact that, as more mobility occurs in the network, the chances of changes in the network structure increase, and, thus, the RSs take more actions. Also, Fig.~\ref{fig:freq} shows that the case of $M=20$~RSs yields an average total number of actions significantly higher than the case of $M=10$~RSs. The reason of this difference is that, as the number of RSs $M$ increases, the possibility of finding new partners when the RSs move increases significantly, hence yielding an increase in the topology variation as reflected by the average total number of actions. In this regard, for $M=20$~RSs, the average total number of actions per minute varies from around $5.7$ at $9$~km/h to around $41$ at $72$~km/h while for $M=10$~RSs, this variation is from  $1.3$ at $9$~km/h to around $12$ at $72$~km/h. In summary, Fig.~\ref{fig:freq} demonstrates how, through periodic runs of the proposed network formation algorithm, the RSs can adapt the topology through appropriate decisions.

\begin{figure}[!t]
\begin{center}
\includegraphics[width=80mm]{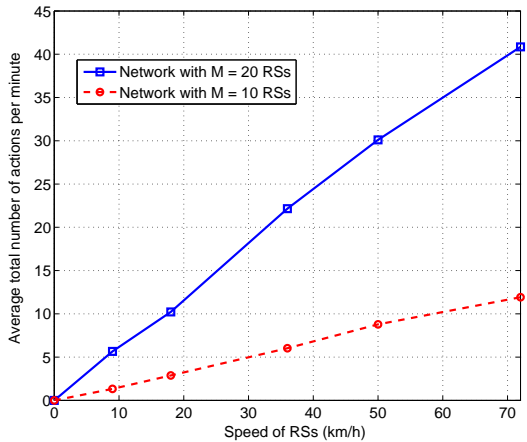}
\end{center}\vspace{-0.65cm}
\caption {Average total number of actions (taken by all RSs) per minute for different RS speeds in networks with different sizes with $40$~MSs.}\vspace{-0.3cm}
\label{fig:freq}
\end{figure}

Fig.~\ref{fig:mobconv} shows how the tree structure in a network with $M=10$~RSs and $40$~MSs, evolves and self-adapts over time when all the MSs are moving at a constant speed of $100$~km/h for a period of $5$ minutes. The proposed network formation algorithm is repeated by the RSs, periodically, every $\theta=30$~seconds, in order to provide self-adaptation to mobility. Fig.~\ref{fig:mobconv} shows that, after $10$ actions taken by the RSs, the network starts with a tree structure with an average number of $2.2$ hops in the tree at time $t=0$. As time evolves, the mobiles are moving and, thus, the RSs engage in the proposed network formation algorithm, to adapt the tree structure to the MSs' mobility through adequate actions. For example, after $2.5$~minutes have elapsed, the tree structure has an average number of $2.6$ hops (after having $1.83$ hops at $2$ minutes), due to the occurrence of a total of $4$ actions by the RSs. At some points such as at $t=4.5$~minutes or $t=5$~minutes, mobility does not yield any changes in the tree structure as no actions are taken by the RSs. Finally, once all the $5$ minutes have passed, the network tree structure is finally made up of an average of $2.5$ hops after a total of $24$ actions played by the RSs during the whole $5$~minutes duration.\vspace{-0.3cm}

\begin{figure}[!t]
\begin{center}
\includegraphics[width=80mm]{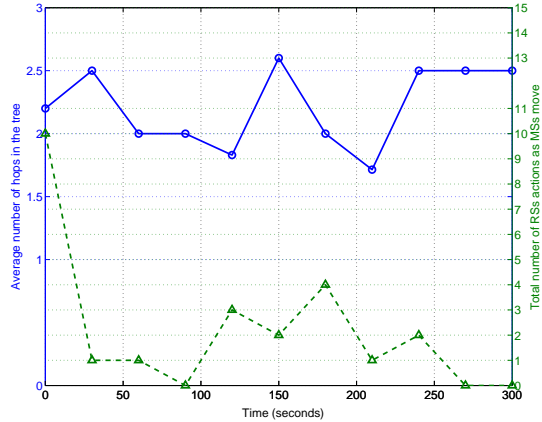}
\end{center}\vspace{-0.6cm}
\caption {Evolution of the network tree structure over time as the MSs are moving with a speed of $100$~km/h over a period of $5$~minutes for a network with $40$~MSs and $M=10$~RSs.} \label{fig:mobconv}\vspace{-0.3cm}
\end{figure}

\section{Conclusions}\label{sec:conc}\vspace{-0.1cm}
In this paper, we have introduced a novel approach for forming the tree architecture that governs the uplink network structure of next generation wireless systems such as LTE-Advanced or WiMAX 802.16j. For this purpose, we formulated a network formation game among the RSs and we introduced a cross-layer utility function that takes into account the gains from cooperative transmission in terms of improved effective throughput as well as the delay costs incurred by multi-hop transmission. To form the tree structure, we devised a distributed myopic algorithm. Using the proposed network formation algorithm, each RS can take an individual decision to optimize its utility by selecting a suited next-hop partner, given the approval of this partner. We showed the convergence of the algorithm to a Nash network structure and we discussed how, through periodic runs of the algorithm, the RSs can adapt this structure to environmental changes such as mobility or incoming traffic. Simulation results demonstrated that the algorithm presents significant gains in terms of average achieved mobile station utility which is at least $21.5\%$ better than the case with no RSs and reaches up to $45.6\%$ improvement compared to a nearest neighbor algorithm.  The results also show that the average number of hops in the tree does not exceed $3$ even for a network with up to $25$ RSs.

\def\baselinestretch{0.88}
\bibliographystyle{IEEEtran}
\bibliography{references}

\end{document}